\documentclass[11pt,reqno]{amsart}

\usepackage{soul,times,amssymb,amsthm,amsmath,amsfonts,bbm,mathrsfs,xcolor,subfigure,graphicx,enumitem,booktabs,euscript,nicefrac}
\allowdisplaybreaks
\usepackage{graphicx}
\usepackage[unicode]{hyperref}
\hypersetup{
    colorlinks,
    linkcolor={blue!80!black},
    citecolor={blue!80!black},
    urlcolor={blue!80!black}
}
\usepackage[normalem]{ulem}
\usepackage[font=footnotesize]{caption}
\usepackage[utf8]{inputenc}
\usepackage[T1]{fontenc}
\usepackage{xargs}
\usepackage[colorinlistoftodos,prependcaption]{todonotes}
\usepackage[left=1in,right=1in,top=1.3in,bottom=1.3in]{geometry}
\newtheorem{theorem}{Theorem}[section]
\newtheorem{lemma}[theorem]{Lemma}

\newtheorem{proposition}[theorem]{Proposition}

\newtheorem{construction}{Construction}[section]
\theoremstyle{remark}
\newtheorem{remark}{Remark}

\newcommand\nc\newcommand
\nc{\bfsl}{\bfseries\slshape}
\nc{\bfit}{\bfseries\itshape}
\nc{\sfsl}{\sffamily\slshape}
\nc{\dfn}{\sffamily\slshape\small}

\nc\bfa{{\boldsymbol a}}\nc\bfA{{\boldsymbol A}}\nc\cA{{\EuScript A}}
\nc\bfb{{\boldsymbol b}}\nc\bfB{{\boldsymbol B}}\nc\cB{{\EuScript B}}
\nc\bfc{{\boldsymbol c}}\nc\bfC{{\boldsymbol C}}\nc\cC{{\mathscr C}}
\nc\bfd{{\boldsymbol d}}\nc\bfD{{\boldsymbol D}}\nc\cD{{\EuScript D}}
\nc\bfe{{\boldsymbol e}}\nc\bfE{{\boldsymbol E}}\nc\cE{{\EuScript E}}
\nc\bff{{\boldsymbol f}}\nc\bfF{{\boldsymbol F}}\nc\cF{{\mathscr F}}
\nc\bfg{{\boldsymbol g}}\nc\bfG{{\boldsymbol G}}\nc\cG{{\EuScript G}}
\nc\bfh{{\boldsymbol h}}\nc\bfH{{\boldsymbol H}}\nc\cH{{\mathcal H}}
\nc\bfi{{\boldsymbol i}}\nc\bfI{{\boldsymbol I}}\nc\cI{{\mathcal I}}
\nc\bfj{{\boldsymbol j}}\nc\bfJ{{\boldsymbol J}}\nc\cJ{{\EuScript J}}
\nc\bfk{{\boldsymbol k}}\nc\bfK{{\boldsymbol K}}\nc\cK{{\EuScript K}}
\nc\bfl{{\boldsymbol l}}\nc\bfL{{\boldsymbol L}}\nc\cL{{\EuScript L}}
\nc\bfm{{\boldsymbol m}}\nc\bfM{{\boldsymbol M}}\nc\cM{{\EuScript M}}
\nc\bfn{{\boldsymbol n}}\nc\bfN{{\boldsymbol N}}\nc\cN{{\EuScript N}}
\nc\bfo{{\boldsymbol o}}\nc\bfO{{\boldsymbol O}}\nc\cO{{\EuScript O}}
\nc\bfp{{\boldsymbol p}}\nc\bfP{{\boldsymbol P}}\nc\cP{{\EuScript P}}
\nc\bfq{{\boldsymbol q}}\nc\bfQ{{\boldsymbol Q}}\nc\cQ{{\EuScript Q}}
\nc\bfr{{\boldsymbol r}}\nc\bfR{{\boldsymbol R}}\nc\cR{{\EuScript R}}
\nc\bfs{{\boldsymbol s}}\nc\bfS{{\boldsymbol S}}\nc\cS{{\EuScript S}}
\nc\bft{{\boldsymbol t}}\nc\bfT{{\boldsymbol T}}\nc\cT{{\EuScript T}}
\nc\bfu{{\boldsymbol u}}\nc\bfU{{\boldsymbol U}}\nc\cU{{\EuScript U}}
\nc\bfv{{\boldsymbol v}}\nc\bfV{{\boldsymbol V}}\nc\cV{{\mathscr V}}
\nc\bfw{{\boldsymbol w}}\nc\bfW{{\boldsymbol W}}\nc\cW{{\mathscr W}}
\nc\bfx{{\boldsymbol x}}\nc\bfX{{\boldsymbol X}}\nc\cX{{\EuScript X}}
\nc\bfy{{\boldsymbol y}}\nc\bfY{{\boldsymbol Y}}\nc\cY{{\mathscr Y}}
\nc\bfz{{\boldsymbol z}}\nc\bfZ{{\boldsymbol Z}}\nc\cZ{{\EuScript Z}}

\nc{\1}{{\mathbbm{1}}}
\nc\mC{{\mathcal C}}
\nc\bmC{{\boldsymbol{\mathcal C}}}
\nc\rr{{\mathbb R}}
\nc\ee{{\mathbb E}}
\nc\sS{{\mathcal S}}
\nc{\integers}{{\mathbb Z}}
\nc{\ff}{{\mathbb F}}
\nc{\ii}{{\mathbb I}}
\nc{\nn}{{\mathbb N}}
\nc{\sC}{{\mathfrak C}}
\nc{\sH}{{\mathfrak H}}
\nc{\sL}{{\mathfrak L}}
\nc\hH{{\mathsf H}}
\nc\gG{{\mathsf G}}
\nc\bfPi{{\boldsymbol{\Pi}}}
\nc{\remove}[1]{}
\nc\Renyi{R{\'e}nyi }
\nc\tbt{\textbullet\quad}
\nc\ffm{\ff_2[x_1, x_2, \dots, x_m]}

\DeclareSymbolFont{bbold}{U}{bbold}{m}{n}
\DeclareSymbolFontAlphabet{\mathbbold}{bbold}

\DeclareMathOperator{\NC}{NC}
\DeclareMathOperator{\BSC}{BSC}

\nc{\genstirlingI}[3]{ \genfrac{[}{]}{0pt}{#1}{#2}{#3}}
\nc{\genstirlingII}[3]{\genfrac{\{}{\}}{0pt}{#1}{#2}{#3}}
\nc{\stirlingI}[2]{\genstirlingI{}{#1}{#2}}
\nc{\dstirlingI}[2]{\genstirlingI{0}{#1}{#2}}
\nc{\tstirlingI}[2]{\genstirlingI{1}{#1}{#2}}
\nc{\stirlingII}[2]{\genstirlingII{}{#1}{#2}}
\nc{\dstirlingII}[2]{\genstirlingII{0}{#1}{#2}}
\nc{\tstirlingII}[2]{\genstirlingII{1}{#1}{#2}}

\nc\todoin[2][]{\todo[inline, caption={2do}, #1]{
		\begin{minipage}{\textwidth-4pt}#2\end{minipage}}}
\nc{\mynote}[2]{{\color{#1} \marginpar{\tiny #2}}}
\nc{\mybignote}[2]{{\color{#1} $\langle \langle$ #2$\rangle \rangle$}}
\newcommandx{\rednote}[2][1=]{\todo[linecolor=red,backgroundcolor=red!25,bordercolor=red,#1]{#2}}
\newcommandx{\bluenote}[2][1=]{\todo[linecolor=blue,backgroundcolor=blue!25,bordercolor=blue,#1]{#2}}
\newcommandx{\yellownote}[2][1=]{\todo[linecolor=yellow,backgroundcolor=yellow!25,bordercolor=yellow,#1]{#2}}
\newcommandx{\greennote}[2][1=]{\todo[inline,linecolor=olive,backgroundcolor=green!25,bordercolor=olive,#1]{#2}}

\nc{\new}[1]{{\color{blue} #1}}
\nc{\rmark}[1]{{\color{red} #1}}
\nc{\pmark}[1]{\textcolor[rgb]{0.50,0.00,1.0}{ #1}}
\nc{\brmark}[1]{\textcolor[rgb]{0.63,0.31,0.2}{ #1}}
\nc{\dgmark}[1]{\textcolor[rgb]{0.00,0.8,0.4}{ #1}}
\nc{\mmark}[1]{{\color{magenta} #1}}
\nc{\ab}[1]{{\footnotesize\color{red}{[AB: #1]}}}
\nc{\mad}[1]{{\footnotesize\color{blue}{[MP: #1]}}}
\nc\redout{\bgroup\markoverwith{\textcolor{red}{\rule[0.5ex]{2pt}{0.8pt}}}\ULon}
\nc\blueout{\bgroup\markoverwith{\textcolor{blue}{\rule[0.5ex]{2pt}{0.8pt}}}\ULon}

\begin{document}

\title{Regular LDPC codes on BMS wiretap channels: Security bounds}
\author[]{Madhura Pathegama}\thanks{The authors are with the Department of ECE and Institute for Systems Research, University of Maryland, College Park, MD 20742. Emails: \{madhura,abarg\}@umd.edu. Research supported the US NSF grant CCF2330909.}
\author[]{Alexander Barg}

\begin{abstract} 
We improve the secrecy guarantees for transmission over general binary memoryless symmetric wiretap channels that relies on regular LDPC codes. 
Previous works showed that LDPC codes achieve secrecy capacity of some classes of wiretap channels while leaking $o(n)$ bits
of information over $n$ uses of the channel. In this note, we improve the security component of these results by reducing the
leakage parameter to $O(\log^2 n)$.  While this result stops short of proving \emph{strong security}, it goes beyond the general 
secrecy guarantees derived from properties of capacity-approaching code families.
\end{abstract}

\maketitle

\section{Introduction}

\subsection{The wiretap channel} Wyner's model of the wiretap channel \cite{wyner1975wire} comprises three terminals, $A$, $B$, and $E$, of which $A$ communicates with $B$ by transmitting a message $\bfM$ selected from a finite set $\cM$  over a channel $\cW_B$, while the eavesdropper $E$ observes the communication through another channel, $\cW_E$. We will assume that both $\cW_B$ and $\cW_E$ are binary-input memoryless symmetric (BMS) channels. A message $\bfM\in\cM$ is encoded into a subset of $\ff_2^n$, and a random bit sequence, $\bfX$, chosen from it is sent from $A$ to $B$ in $n$ uses of the channel $\cW_B.$ Terminal $B$ observes the sequence $\bfY$ while terminal $E$ observes the sequence $\bfZ$, obtained on the output of $\cW_E$. The encoding is constructed to support reliable transmission to $B$ and at the same time to avoid revealing information to $E$. The reliability requirement amounts to the condition $\Pr(\bfM \neq \hat{\bfM}) \to 0$ as $n\to\infty$, where $\hat{\bfM}$ is the estimate of $\bfM$ made by $B$, and where we assume some prior distribution on $\cM$. To ensure secrecy in communication, the {\em information leakage} $I(\bfM; \bfZ)$ must be small. 
Wyner \cite{wyner1975wire} defined the secrecy condition as $\frac{1}{n} I(\bfM; \bfZ) \to 0$. Maurer \cite{maurer1994strong} observed that from the cryptographic perspective, a more meaningful requirement is $ I(\bfM; \bfZ) \to 0$ without normalizing by the number of channel uses.
The normalized condition is now referred to as {\em weak secrecy}, while the unnormalized condition is known as {\em strong secrecy}. Denote by $R=\frac 1n\log|\cM|$ the transmission rate. The {\em secrecy capacity} $C_s(\cW_B, \cW_E)$ is defined as the supremum of the rates that permit reliable transmission and
support the secrecy condition. The value of $C_s(\cW_B, \cW_E)$ does not depend on the type of secrecy
used in the assumptions, in other words, it does not decrease by moving from weak to strong secrecy. We refer to \cite[Sec. 17.2]{csiszar2011information} for a formal discussion of
the $(\cW_B,\cW_E)$ wiretap channel model. 

\subsection{Capacity, secrecy, and constructive coding schemes} Designing transmission schemes for the wiretap channel $(\cW_B,\cW_E)$ faces two separate but related tasks, supporting reliable communication $A\to B$ and hiding the message from $E$. In this paper we focus on the second of these problems, examining the secrecy component of transmitting with regular LDPC codes. Early works showed that random coding achieves secrecy capacity of the wiretap channel under both weak and strong security \cite{wyner1975wire,csiszar1978broadcast}. 
It is more difficult to find constructive coding schemes that match this performance. With few exceptions (e.g., \cite{hayashi2010construction,bellare2012semantic,cheraghchi2011invertible}), such schemes rely on Wyner's {\em coset coding scheme}. This scheme is constructed using two linear codes, $\mC_E$ and $\mC_B$, such that $\mC_E \subset \mC_B$ and $\cM = \mC_B/\mC_E$. In this {\em nested construction}, each message $m$ is encoded into a coset in $\mC_B/\mC_E$, and the sequence transmitted by terminal $A$ is chosen as a uniform random vector from the coset. 

Message recovery at terminal $B$ is ensured by designing the code $\mC_B$ that supports reliable transmission the channel $\cW_B$. This requires that the rate of $B$ satisfy $R(\mC_B) \leq C(\cW_B)$.
On the other hand, as shown in \cite{bloch2011achieving,luzzi2011capacity}, achieving strong secrecy requires that $R(\mC_E) \geq C(\cW_E)$.
Let us emphasize that it is not enough to simply choose a code $\mC_E$ of such rate, what is needed is a close approximation of the output distribution of $\cW_E$ that corresponds to the uniform distribution on its input \cite{bloch2013strong}.
This property is also known as {\em channel resolvability} \cite{han1993approximation}, and it has found applications in the context of secret communication; see the cited works as well as \cite{hayashi2006general}. 
The smallest code rate that permits resolvability is a characteristic of the channel, called resolvability threshold.
By a fundamental result of \cite{han1993approximation}, if closeness to the capacity-achieving output distribution is measured by the KL divergence, this threshold equals Shannon capacity of the channel, and it is achieved by random codes. 
Later works have studied resolvability thresholds for the case when KL divergence is replaced by \Renyi divergence $D_\alpha$ of 
order $\alpha>1$, representing a more stringent constraint; see
\cite{yu2018renyi,pathegama2023smoothing} for recent work on \Renyi\!-based secrecy. 
If $\cW_E$ is a binary symmetric channel (BSC), the capacity-achieving
output distribution is also uniform. 
In this case resolvability is often called {\em code smoothing}
\cite{pathegama2023smoothing}, inspired by the term adopted in cryptographic applications \cite{debris2023smoothing,pathegama2024LPN}.

Most known constructive transmission schemes for the wiretap channel rely on LDPC and polar codes, which have a natural nested structure. While weak secrecy can be achieved based on the capacity -achieving properties of the code
ensembles \cite{andersson2010nested,hof2010secrecy,thangaraj2007applications,rathi2009two}, for strong secrecy one needs to show the smoothing (more precisely, resolvability) property  of the code family, which is often a nontrivial task. For polar codes on the wiretap channel this problem was sidestepped in \cite{gulcu2016achieving,chou2016polar,mahdavifar2011achieving} by allowing terminals $A$ and $B$ to share
a vanishing amount of information, and thus departing from the transmission scheme considered here. 
At the same time, the situation is less well understood for LDPC codes, where only partial results are available. In particular, 
the authors of \cite{thangaraj2007applications} showed that it is possible to approach channel capacity with weak security when both $\cW_B$ and $\cW_E$ are binary erasure channels channels (BECs). 
Their proof leveraged the inherent (weak) resolvability-achieving properties of capacity-achieving codes. Several other results for LDPC codes were obtained in \cite{rathi2009two,Subramanian2010,rathi2012performance}. Schemes based on {\em dual} LDPC codes were employed for 
the wiretap channel with noiseless channel $\cW_B$ and a BEC $\cW_E$ \cite{thangaraj2007applications,Suresh2010,subramanian2011strong}.

Unlike previous constructions based on nested LDPC codes, which focused mostly on erasure wiretap channels, we establish security guarantees for general BMS wiretap channels. When specialized to the erasure case, our bounds are tighter than the results available in the 
literature.

\vspace*{-.1in}
\subsection{Our results} 
We show that for BMS wiretap settings
$(\cW_B, \cW_E)$, it is possible to design nested binary LDPC codes with rates close to $C(\cW_B) - C(\cW_E)$, (given that $C(\cW_B) > C(\cW_E)$) while ensuring reliable communication and limiting the information leakage to $I(\bfM; \bfZ) = O(\log^2 n)$. If $\cW_B$ is {\em less noisy} than $\cW_E$ 
(see \cite{korner1975comparison,csiszar2011information} for the definition), the author of \cite{van1997special} 
showed that $C_s(\cW_B,\cW_E) = C(\cW_B) - C(\cW_E)$. In such case, our results show that for rates just below the wiretap capacity, it is possible to achieve reliability and small leakage using nested LDPC codes.

For comparison, we note that natural properties of capacity-achieving code families used on the wiretap channel result in information leakage
on the order of $o(n)$. This can be improved to $O(\sqrt n)$ with the second-order results for resolvability that parallel the
results of \cite{polyanskiy2010channel} for communication. Breaking through this barrier requires reliance on specific properties
of codes rather than general theorems. In this paper we prove a tighter estimate for resolvability of LDPC codes based on the
properties of the regular ensemble. This result still stops short of proving strong secrecy, but represents an improvement over the previously known general bounds. 

It suffices to consider the case where eavesdropper's channel $\cW_E$ is a BSC because it extends to the general case using channel comparison results \cite[Prob. 6.18]{csiszar2011information},  \cite{csacsouglu2011polar}.
The BSC case is addressed by assuring that the marginal distribution induced by any coset of $\mathcal{C}_E$ on the output of $\cW_E$ is nearly uniform. 
Formally this amounts to showing that with high probability, the KL divergence satisfies $D(P_{\bfZ\|\bfM=m} \| P_{\bfU_n}) = O(\log^2 n)$, where $P_{\bfU_n}$ is a uniform distribution on $\ff_2^n$.
To achieve this result, we prove a stronger bound by showing that the \Renyi divergence of order $\alpha$ satisfies $D_\alpha(P_{\bfZ|\bfM=m} \| P_{\bfU_n}) = O(\log^2 n)$, where $\alpha > 1$ is sufficiently close to 1. In terms of
the above discussion, we analyze smoothing properties of LDPC codes relying on the \Renyi divergence.

To implement this program, we start with a BMS channel $\cW_B$ and a BSC $\cW_E$.
We set $\mC_B$ to be an LDPC code from Gallager's ensemble with rate slightly below the capacity $C(\mathcal{W}_B)$ 
and choose another LDPC code $\mathcal{C}_E \subset \mathcal{C}_B$ with rate slightly above $C(\cW_E)$. We then show, in terms of \Renyi divergence bounds,
that the random syndrome of the LDPC code is approximately uniform, which amounts to the secrecy claim.

Since we focus on the security part of the transmission scheme and do not address efficient decoding, we
did not attempt to analyze uniformity properties of irregular LDPC codes which attain channel capacity
under belief propagation decoding, which remains an open question. 
With regard to the analysis of the information leakage, we make use of the \Renyi divergence $D_\alpha$, which is a new element in the derivation of the bound. Previously \Renyi divergence was utilized in \cite{hayashi2011exponential} for random codes, while here we use it
for structured codes or code ensembles.

\section{Preliminaries}\label{sec:preliminaries}

In this work, random variables are denoted by capital boldface letters. For a finitely supported random variable $\bfV$, we denote its probability mass function by $P_\bfV$. If $\bfV$ follows a probability distribution $P$, we write $\bfV \sim P$ to indicate that $P_\bfV = P$. If $\bfV$ is uniformly distributed over a finite set $\cA$, we denote this by $\bfV \sim_U \cA$ 
or 
$\bfV \sim P_\cA$. 
For two independent random vectors $\bfV,\bfV'$ in $\ff_2^n$, the distribution of their sum
is given by a convolution 
    $$
    P_{\bfV + \bfV'}(v) =( P_{\bfV}\ast P_{\bfV'})(v) = \sum_{v' \in \ff_2^n}P_{\bfV}(v-v')P_{\bfV'}(v').
    $$  
A BMS channel $\cW:X=\{0,1\}\to Y$ is a stochastic mapping from the two-element set to an output alphabet $Y$ which can be finite or infinite. In the finite case, $\cW$ amounts to a conditional probability distribution $P(y|x), y\in Y, x\in X$ such that the $2\times |Y|$ matrix $P(y|x)$ has certain permutation invariance \cite[p.178]{richardson2008modern}. For $Y=\rr$ and continuous distributions, these conditions are replaced with pdf's. Denote the Shannon capacity of the channel by $C(\cW)$, and denote by $P_{\cW}^{\text{ML}}(\mC)$ the error probability of maximum likelihood
decoding of a binary code $\mC$ used for transmission over $\cW$.

Below the symbols $\bfM,\bfX,\bfY$ and $\bfZ$ are reserved to denote the uncoded message, the coded message, and the observations of $B$ and $E$, respectively.
For the wiretap channel $(\cW_B,\cW_E)$ and a coset coding scheme $(\mC_B,\mC_E$), denote by $P_{(\cW_B,\cW_E)}^{\text{ML}}(\mC_B/\mC_E)$ the error probability of estimating the message (coset of $\mC_B/\mC_E$) by $B$ from the observation $\bfY$.
Clearly,  $P_{(\cW_B,\cW_E)}^{\text{ML}}(\mC_B/\mC_E)\le P_{\cW_B}^{\text{ML}}(\mC_B)$, the error probability of identifying the nearest codeword
of $\mC_B$. 
Denote by $$L(\mC_B,\mC_E):=I(\bfM;\bfZ)$$ the information leakage in this setting. Denote the rate of message transmission in the
nested scheme by $R(\mC_B,\mC_E):=R(\mC_B)-R(\mC_E)$, where $R(\cdot)$ is the rate of the code argument.

\subsection{\Renyi divergences} For two discrete distributions $P$ and $Q (P \ll Q)$ defined on the set $\cX$, and for $\alpha \in (1, \infty)$, the \Renyi divergence is defined as
\begin{align}
    D_{\alpha}(P\|Q)=\frac{1}{\alpha -1} \log \sum_x P(x)^{\alpha} Q(x)^{-(\alpha -1)},
\end{align}
where the logarithm is taken base 2.
In the limit of $\alpha\to 1$, $D_{\alpha}$ becomes the KL divergence $D(P\|Q)$.
Note that $D_\alpha$ is monotone increasing, i.e., for $1\le \alpha<\alpha'$ we have $D_\alpha(P\|Q)\le D_{\alpha'}(P\|Q)$. 

If $Q$ is uniform, then $D_\alpha(P\|Q)=\log|\cX|-H_{\alpha}(P)$, where $H_\alpha(P)$ is the \Renyi entropy of order $\alpha$, defined as:
\begin{align}
   H_{\alpha} (P)&=\frac1{1-\alpha}\log\Big(\sum_x P_X(x)^{\alpha} \Big) \quad\quad \alpha \in(1,\infty ). \label{eq:Renyi}
 \end{align}  
For $\alpha = 1$, $H_{\alpha} (P)$ reduces to the Shannon entropy $H(P)$.
When $P$ is supported on two points, i.e. $P = (p,1-p)$ for some $p\in [0,1]$, we use the notation $h_\alpha(p)$ and $h(p)$ for the \Renyi and Shannon entropies.

\subsection{LDPC codes} 
To define Gallager's LDPC ensemble $\cG(n,R,s)$ \cite{gallager1963low}, we fix the code length $n$, rate $R$, and the row weight $s$ (the number of 1's in each row of the parity-check matrix $\bfH$).

We will describe the construction of $\bfH$. Let $(F_i)_{i=1}^{n/s}$ be a collection of binary matrices of size $n/s \times s$, where each $F_i$ has ones in $i$-th row and zeros elsewhere, and let
$ F \in \ff_2^{(n/s) \times n} $ be given by
\begin{align}\label{eq: Fm}
    F = [F_1 \mid F_2 \mid \dots \mid F_{n/s}].
\end{align}
Next, choose $(1-R)s$ independent random $(n/s\times n)$ permutation matrices $\bfPi_i$ and let
$\bfH_i=F\bfPi_i, i=1,\dots,(1-R)s$. Finally, let $\bfH=(\bfH_1^\intercal,\dots,\bfH_{(1-R)s}^\intercal)^\intercal$ be the target $(1-R)n\times n$ parity-check matrix.
\begin{remark}
The value of $s$ is related to the gap to channel capacity when the codes are used on a BMS channel. Bounds on the density of ones in $\bfH$ as a function of gap to the channel capacity for belief propagation decoding have been studied in detail, see, for instance, \cite{richardson2008modern}, Thm.4.150. To claim our results, we will need to take $s=O(\log n)$. 
\end{remark}

\section{LDPC codes on the BMS-BSC wiretap channel}

In this section, which forms the main technical part of this work, we show that LDPC codes approach
secrecy capacity of  BMS-BSC wiretap channels $(\cW_B,\cW_E)$ with $C(\cW_E)<C(\cW_B)$. 
Under this approach, we show that nested Gallager's LDPC codes support reliable transmission over the
wiretap channel at rates close to $C(\cW_B)-C(\cW_E)$ with $L(\mC_B,\mC_E)=O(\log^2 n$). 

\subsection{Plan of the proof} Before presenting our results, we provide some background used to explain our motivation. 
Let $P_{\mC_E}$ be the (uniform) distribution on the codewords of $\mC_E$
and let $\cW_E^n\circ P_{\mC_E}$ be the distribution induced at the output (of $n$ uses)
of the channel $\cW_E$. The following bound for information leakage is well known (e.g., \cite[Lemma 3]{pathegama2023smoothing}):
\begin{align}\label{eq: leakage}
    I(\bfM;\bfZ) \leq D(\cW_E^n \circ P_{\mC_E}\|P_{\bfU_n}).
\end{align}

Suppose that $\cW_E=\text{BSC}(p)$, then the input is acted upon by additive Bernoulli$(p)$ noise. 
Denote by $\beta(p)$ the $n$-fold product $\text{Bern}(p)^{\otimes n}$, so that $\cW_E^n \circ P_{\mC_E} = P_{\bfX' + \bfV}$ where $\bfX' \sim_U \mC_E$, and 
$\bfV \sim \beta(p)$. 
Next, we seek to relate the quantity $D(\cW_E^n \circ P_{\mC_E}\|P_{\bfU_n})$ to the parity-check
matrix (which later will define the LDPC code $\mC_E$). The answer is provided in the following lemma, which forms a special (limiting)  case of Lemma 3.6 in \cite{pathegama2024r}. For completeness, we give a self-contained proof.
\begin{lemma}\label{lem:sm_to_proj}
   Let $\mC$ be an $[n,k]$ binary linear code and let $H$ be its parity check matrix. Let $\bfX'$ be a uniform random codeword of $\mC$ and $ \bfV$ be a random vector in $\ff_2^n$. Then 
   \begin{align*}
       D(P_{\bfX' + \bfV}\|P_{\bfU_n})=D( P_{H \bfV} \|P_{\bfU_{n-k}}).
   \end{align*} 
\end{lemma}
\begin{proof}    For any $c \in \mC$,
    $$ P_{\bfX'+\bfV}(v) = P_{\mC}\ast P_\bfV (v) = P_{\mC}\ast P_\bfV (v+c) = P_{\bfX'+\bfV}(v+c).$$
Using this equality, we have 
    \begin{align*}
        P_{H \bfV}(u) = P_{H(\bfX'+\bfV)}(u) = \sum_{v: H v = u}P_{\bfX'+\bfV}(v)
        = |\mC|P_{\bfX'+\bfV}(v_u),
    \end{align*}
where $v_u$ is a representative of the coset defined by $u$. Denote by $\bfU_m$ the uniform distribution 
over $\ff_2^m$. We have
    \begin{align*}
        D(P_{H \bfV}\|P_{\bfU_{n-k}}) 
        &= \sum_{u \in \ff_2^{n-k}} P_{H \bfV}(u)\log \Big(\frac{P_{H \bfV}(u)}{P_{\bfU_{n-k}}(u)}\Big)\\
        &= \sum_{u \in \ff_2^{n-k}} |\mC|P_{\bfX'+\bfV}(v_u)\log \Big(\frac{|\mC|P_{\bfX'+\bfV}(v_u)}{P_{\bfU_{n-k}}(u)}\Big)\\
        &= \sum_{u \in \ff_2^{n-k}} |\mC|P_{\bfX'+\bfV}(v_u)\log \Big(2^nP_{\bfX'+\bfV}(v_u)\Big)\\
        &= \sum_{u \in \ff_2^{n-k}}\sum_{v\in \ff_2^n: H v =u} P_{\bfX'+\bfV}(v)\log \Big(2^nP_{\bfX'+\bfV}(v)\Big)\\
        &= \sum_{v \in \ff_2^{n}} P_{\bfX'+\bfV}(u)\log \Big(\frac{P_{\bfX'+\bfV}(v)}{P_{\bfU_n}(v)}\Big)\\
        &=D(P_{\bfX' + \bfV}\|P_{\bfU_n}). \qedhere
    \end{align*}
\end{proof}
With this, it is sufficient to provide an upper bound for $D( P_{H \bfV} \|P_{\bfU_{n-k}})$, where $\bfV \sim \beta(p)$ and $H$ is a LDPC matrix. Since our codes are sampled from a random ensemble, all the above arguments must hold with high probability. To establish this, we first show that the  expected leakage is at most 
$O(\log^2 n)$ and then apply Markov's inequality. 
In mathematical terms, we seek to prove a bound of the form 
   \begin{equation}\label{eq: ED}
   \ee_\bfH[D( P_{\bfH \bfV} \|P_{\bfU_{n-k}}|\bfH)] = O(\log^2 n).
   \end{equation}
or equivalently, the relation  $D(P_{\bfH \bfV, \bfH}\|P_{\bfU_{n-k}}P_\bfH) = O(\log^2 n).$

\subsection{Technical results} To prove a bound of the form \eqref{eq: ED}, we find it 
easier to estimate the \Renyi divergence $D_\alpha$ rather than the KL divergence. The resulting inequality, written in \eqref{eq: main} in exponential form, represents the main technical contribution of this work. 
Before stating it, let us introduce some notation. Denote by $B_t(x)$ the metric
ball in $\ff_2^n$ with center at $x$ and radius $t$ and let $\nu_t:=|B(t)|$ be its volume (independent of the center). Let $S_t=B_t(0)\backslash B_{t-1}(0)$ and let $\Gamma(x,t)=B_t(x)\cup B_t(x+{\bf 1}),$ where ${\bf 1}$ is the all-ones vector.
\begin{proposition}\label{prop: main}
    Let $\alpha \in (1,2)$ and $t \in [n]$. Let  $\mathcal H$ be a set of $\ff_2^{m \times n}$ matrices. Let $\bfH$ be a uniform random matrix from $\mathcal H$.
    Suppose $\bfV$ is a random vector from $\ff_2^n$ independent of $\bfH$. Then 
    \begin{align}\label{eq: main}
        2^{(\alpha-1)D_{\alpha}(P_{\bfH \bfV, \bfH}\|P_{\bfU_m}P_\bfH)} \leq 2\nu_t & 2^{(\alpha-1)(m-H_\alpha(P_{\bfV}))}\nonumber\\
        & + \Big(2^m\sum_{v \in \ff_2^n \setminus \Gamma(0,t)}\Pr(\bfH v = 0) P_\bfV \ast P_\bfV(v)\Big)^{\alpha-1}.
    \end{align}
\end{proposition}
\vspace*{.1in} For a given BSC $\cW$, define $\bfH \sim \cG(n,R, a\log n)$ for some $a >0$ and $R = C(\cW) + \epsilon$,  $\epsilon >0$. Our next goal is to derive an upper bound for the quantity $D(P_{\bfH \bfV, \bfH}\|P_{\bfU_m}P_\bfH)$. To achieve this, we first need an estimate for  $\Pr(\bfH v = 0)$ 
The following lemma, which is a minor modification of  \cite[Lemma 2.13]{mosheiff2021low},  addresses this goal. 
For completeness, we provide a proof.
\begin{lemma}\label{lem: Hv}
    Let $v \in \ff_2^n$ be a fixed vector of weight $w \geq 1$. Let $\bfH \sim \cG(n,R,s)$. Then
    \begin{align*}
        \Pr(\bfH v = 0)  \leq (\sqrt{2n})^{(1-R)s}\Big(\frac{1+(1-\frac{2w}{n})^s}{2}\Big)^{n(1-R)}.
    \end{align*}
\end{lemma}
\begin{proof}
Let $\{\bfPi_i, i=1,\dots,(1-R)s\}$ denote the collection of independent random $n \times n$ permutation matrices that correspond to the ``layers'' of $\bfH$. The probability that $v\in\ker\bfH$ is
\begin{align*}
    \Pr_{\bfH}(\bfH v = 0) 
    & = \prod_{i=1}^{(1 - R)s} \Pr_{\bfPi_i}(F \bfPi_i v = 0 ) \qquad\text{(see \eqref{eq: Fm})} \\
    & = \Pr_{\bfy \sim_U S_w}(F \bfy = 0 )^{(1-R)s} \\
    & = \Pr_{\bfy \sim \beta_{w/n}}(F \bfy = 0 | |\bfy|=w)^{(1-R)s}\\
    & = \Big(\frac{\Pr_{\bfy \sim \beta_{w/n}}(F \bfy = 0 , |\bfy|=w)}{\Pr_{\bfy \sim \beta_{w/n}}(|\bfy|=w)}\Big)^{(1-R)s}\\
    & \leq \Big(\frac{\Pr_{\bfy \sim \beta_{w/n}}(F \bfy = 0 )}{\Pr_{\bfy \sim \beta_{w/n}}(|\bfy|=w)}\Big)^{(1-R)s}.
\end{align*}

We estimate the numerator and the denominator separately below. Observe that
\begin{align*}
     \Pr_{\bfy \sim \beta_{w/n}}((F \bfy)_1 = 0 ) = \Pr_{\bfy \sim \beta_{w/n}}(\bfy_1 \oplus \bfy_2 \oplus \dots \oplus \bfy_s = 0) = \frac{1+(1-\frac{2w}{n})^s}{2}.
\end{align*}
Therefore we have 
\begin{align*}
    \Pr_{\bfy \sim \beta_{w/n}}(F \bfy = 0 ) = \Big(\frac{1+(1-\frac{2w}{n})^s}{2}\Big)^{n/s}.
\end{align*}
Using Stirling's approximation, we bound the denominator as follows: 
\begin{align*}
    \Pr_{\bfy \sim \beta_{w/n}}(|\bfy|=w) = \binom{n}{w}(w/n)^{w}(1-w/n)^{(n-w)} \geq \sqrt{\frac{n}{8 w (n-w)}}.
\end{align*}
Continuing the calculation,
\begin{align*}
    \Pr_{\bfH}(\bfH v = 0) & \leq \Big(\sqrt{\frac{8 w (n-w)}{n}}\Big(\frac{1+(1-\frac{2w}{n})^s}{2}\Big)^{n/s}\Big)^{(1-R)s}\\
    & \leq (\sqrt{2n})^{(1-R)s}\Big(\frac{1+(1-\frac{2w}{n})^s}{2}\Big)^{n(1-R)}. \qedhere
\end{align*}

\end{proof}

\begin{theorem}\label{thm: smoothing}
    Let $\cW$ be a BSC(p) and let $\bfV\sim\beta(p)$ be a random vector. Then for any $\epsilon> 0$, there exists $a = a(\epsilon,\cW)>0$ such that for $\bfH \sim \cG(n,C(\cW)+\epsilon,a\log n)$, 
    \begin{align*}
        D(P_{\bfH \bfV, \bfH}\|P_{\bfU_{n-k}}P_\bfH) \leq a \log^2 n,
    \end{align*}  
    for sufficiently large $n$.
\end{theorem}

\begin{proof}
    Let $\cW$ be  a $\BSC(p)$ channel. Fix $\epsilon$ and choose $\alpha$ such that $h_\alpha(p) = h(p)-\frac{\epsilon}{2}$. Then choose $\tau \in (0,1/2)$ such that $h(\tau) < \frac{\epsilon}{2}(\alpha-1)$.
We will use the bound in Proposition \ref{prop: main}, setting $\mathcal{H} = \cG(n,R,a\log n)$, $m = n(1-R)$ and $\bfV \sim \beta(p)$. Here $R = C(\cW)+\epsilon$ and $a>0$ is a constant to be fixed later. 
Let us bound the terms on the right-hand side of \eqref{eq: main}, starting with the first one:
    \begin{align}
        2\nu_t2^{(\alpha-1)(m-H_\alpha(P_{\bfV}))} 
        &\leq 2 \times 2^{nh(\tau)+(\alpha-1)(n(1-R)-H_\alpha(P_\bfV))}\nonumber
    \end{align}
Since,  $R = 1-h(p)+\epsilon$ and $H_\alpha(P_\bfV) = nh_\alpha(p)$,
    \begin{align}
        2\nu_t2^{(\alpha-1)(m-H_\alpha(P_{\bfV}))} 
        & \leq  2^{nh(\tau)+(\alpha-1)n(1-(1-h(p)+\epsilon)-h_\alpha(p))+1}\nonumber\\
        & \leq 2^{nh(\tau)-\frac{\epsilon}{2}(\alpha-1)n+1} \leq 2 \label{eq: T1}.
    \end{align}
    Next we bound the second term of the right-hand side of \eqref{eq: main}. To do that let us obtain a simplified bound for the quantity $\Pr(\bfH v = 0)$. Let $|v|\in [\tau n, (1-\tau)n] $. Then from Lemma \ref{lem: Hv},
    \begin{align*}
        \Pr(\bfH v = 0)  
        &\leq (\sqrt{2n})^{(1-R)s}\Big(\frac{1+(1-\frac{2|v|}{n})^s}{2}\Big)^{n(1-R)}\\
        & \leq (\sqrt{2n})^{(1-R)s}\Big(\frac{1+(1-2\tau)^s}{2}\Big)^{n(1-R)}.
    \end{align*}
    Recall that $s = a\log n$. Choosing $a$ such that $a\log(1/(1-2\tau)) \geq 1$, we have 
    \begin{align}
        \Pr(\bfH v = 0)  
        &\leq (\sqrt{2n})^{(1-R)a\log n}\Big(1+\frac{1}{n}\Big)^{n(1-R)}2^{-n(1-R)}\nonumber\\
        & \leq e(\sqrt{2n})^{(1-R)a\log n}2^{-n(1-R)}  \label{eq: T2}.
    \end{align}
 Using bounds \eqref{eq: T1} and \eqref{eq: T2} in \eqref{eq: main}, we obtain
    \begin{align*}
        &2^{(\alpha-1)D_{\alpha}(P_{\bfH \bfV, \bfH}\|P_{\bfU_m}P_\bfH)} \\
         &\leq 2 + \Big(2^{n(1-R)}\sum_{v \in \ff_2^n \setminus \Gamma(0,t)} \Pr(\bfH v = 0) P_\bfV \ast P_\bfV(v)\Big)^{\alpha-1}\\
        & \leq 2 + \Big(2^{n(1-R)}\hspace*{-.1in}\max_{v \in \ff_2^n \setminus \Gamma(0,t)}\hspace*{-.1in}
        \Pr(\bfH v = 0) 
        \hspace*{-.2in}\sum_{v' \in \ff_2^n \setminus \Gamma(0,t)}\hspace*{-.1in} P_\bfV \ast P_\bfV(v')\Big)^{\alpha-1}\\
        & \leq 2 + (e(\sqrt{2n})^{(1-R)a\log n})^{(\alpha-1)}\\
        & \leq n^{(1-R)(\alpha-1)a\log n} \quad (\text{for sufficiently large }n)\\
        & \leq 2^{(\alpha-1)a\log^2 n}. 
    \end{align*}
    This implies that for large $n$, $D_{\alpha}(P_{\bfH \bfV, \bfH}\|P_{\bfU_m}P_\bfH) \leq a\log^2 n.$
\end{proof}

The implication of this theorem is that using LDPC codes of rate just above capacity of a BSC, we can induce an approximately uniform distribution at the output. We use this fact to show that the following nested LDPC ensemble achieves reliability and low leakage at rates close to the capacity. 
\begin{construction}\label{con: nldpc}
    Given $n \in {\mathbb N}$, rate values $R_B$, $R_E \in (0,1)$, and $s \in\{1,\dots,n\}$, the random code pair $\NC(n,R_B,R_E,s) = (\mC_B, \mC_E)$ is constructed as follows.
    Let $\bfH_B \sim \cG(n,R_B, s)$ and $\bfH_{E\backslash B} \sim \cG(n,1-R_B+R_E, s)$. Let   
    $\mC_B=\ker\bfH_B$ and $\mC_E =\ker \bfH_E$, where $\bfH_E = \Big[\frac{\bfH_B}{\bfH_{E\backslash B}}\Big]$. 
\end{construction}

It is straightforward that $\mC_E\subset \mC_B$. Using this code pair in the coset code scheme, we obtain
the result for the BMS-BSC wiretap channel.

\begin{theorem}\label{thm: wtc}
    Let $\epsilon >0 $ and $\delta \in (0,1)$. Let $\cW_B$ be a BMS channel and $\cW_E$ be a BSC satisfying $C(\cW_E) < C(\cW_B)$. 
    There exists $a' = a'(\epsilon, \cW_E)$ such that random codes from the ensemble $ (\mC_B,\mC_E) = \NC(n,C(\cW_B)-\epsilon/4, C(\cW_E)+\epsilon/4,a'\log n)$ satisfy the following with probability at least $1-\delta$: 
    \begin{enumerate}
        \item[$(1)$] $\liminf_{n \to \infty} R(\mC_B,\mC_E) \geq C(\cW_B)-C(\cW_E) -\epsilon$\\[-.1in]
        \item[$(2)$] $\lim_{n \to \infty} P_{(\cW_B,\cW_E)}^{\text{ML}}(\mC_B/\mC_E) = 0,$ and\\[-.0in]
        \item[$(3)$] $\limsup_{n \to \infty}  \frac{L(\mC_B,\mC_E)}{\log^2 n} \leq \frac{2a'}{\delta}$.
    \end{enumerate}  
\end{theorem}

\begin{proof}
    From the well-known properties of LDPC codes \cite{gallager1963low}, the following conditions are satisfied with probability $1-o(1)$.
    \begin{itemize}
        \item $R(\mC_B) > C(\cW_B)-\epsilon/2$ , $R(\mC_E) < C(\cW_B)+\epsilon/2$\\[-.1in] 
        \item $P_{\cW_B}^{\text{ML}}(\mC_B) = o(1)$. 
    \end{itemize}
    Therefore, with probability $1-\delta/2$ both the first and second conditions stated in the theorem are jointly satisfied. Let $a' = a(\epsilon/4,\cW_E)$, where $a$ is the constant defined in Theorem~\ref{thm: smoothing}.
  Putting $m=n(1-R(\mC_E))$ and using the Markov inequality,
    \begin{align*}
        \Pr_\bfH(D(P_{\bfH \bfV}\|P_{\bfU_m}|P_\bfH)\geq 2a'\log^2 n/\delta)
            \leq \frac{D(P_{\bfH \bfV, \bfH}\|P_{\bfU_m}P_\bfH)}{2a'\log^2 n/\delta} = \frac{\delta}{2}.
    \end{align*}
    We have shown that, with probability at least $1-\delta/2$, the third condition is also satisfied. 
    Combining these two claims, we conclude that all three conditions jointly hold with probability at least $1-\delta$. 
\end{proof}


\section{LDPC codes in BMS-BMS wiretap channels}
In this section we extend the BMS-BSC channel results to the general case relying on channel comparison theorems.

\begin{lemma}  
    Let $\cW_B$ and $\cW_E$ be two BMS channels whose capacities satisfy $C(\cW_E) < C(\cW_B)$. Let $\cV_E$ be a BSC with the same capacity as $\cW_E$.
    Let $(\mC_B,\mC_E)$ be a coset code designed for the wiretap channel $(\cW_B,\cW_E)$. Then  
    $$
        L(\mC_B,\mC_E) \leq  2D(\cV_E^n \circ P_{\mC_E}\|P_{\bfU_n}).
    $$
\end{lemma}

\vspace*{.05in}\begin{proof}
   Let $\bfM$ and $\bfX$ be the uncoded and coded versions of the message, i.e., the coset index and a random word from it. Let $\bfZ$ and $\bfZ'$
   be the observations of $E$ on the output of channels $\cW_E^n$ and $\cV_E^n$, respectively. Let us obtain an upper bound for $I(\bfM;\bfZ)$:
    \begin{align*}
        I(\bfM;\bfZ) = I(\bfX;\bfZ)- I(\bfX;\bfZ|\bfM) 
        \leq nC(\cW_E) - I(\bfX;\bfZ|\bfM).
    \end{align*}
 
    Now we proceed to show that $I(\bfX;\bfZ|\bfM) \geq I(\bfX;\bfZ'|\bfM)$.
    Let $\bfX_m \sim P_{\bfX|\bfM=m}$ and let $\bfZ_m$ and $\bfZ'_m$ be the corresponding output random vectors over the channels $\cW_E$ and $\cV_E$ respectively. 
     Since the triple 
    $\bfM\!\mathrel{\vcenter{\hbox{\rule{3mm}{.8pt}}}}\bfX\mathrel{\vcenter{\hbox{\rule{3mm}{.8pt}}}}\bfZ$ satisfies the Markov condition, we obtain 
\begin{align*}
        P_{\bfX,\bfZ|\bfM}(x,z|m) = P_{\bfZ|\bfX}(z|x)P_{\bfX|\bfM}(x|m) = P_{\bfZ_m|\bfX_m}(z|x)P_{\bfX_m}(x) = P_{\bfX_m,\bfZ_m}(x,z),
\end{align*}
where $P_{\bfX_m}(x) := P_{\bfX|\bfM}(x|m)$ and $P_{\bfZ_m|\bfX_m}(z|x) := \cW_E^n(z|x)$.
       Recall that BSCs are the least capable\footnote{
A channel $\cW:X\to Y$ is more capable than a channel $\cV:X\to Z$ if $I(X;Y)\ge I(X,Z)$ for all
joint distributions $P_{\bfX\bfY\bfZ}=P_{\bfX}P_{\bfY|\bfX}P_{\bfZ|\bfX}$.}        
       among all channels with the same capacity \cite[Lemma 7.1]{csacsouglu2011polar}. Consequently
    \begin{align*}
        I(\bfX;\bfZ|\bfM) &= \sum_{m}P_\bfM(m)I(\bfX;\bfZ|\bfM=m) = \sum_{m}P_\bfM(m)I(\bfX_m;\bfZ_m)\\
        & \geq \sum_{m}P_\bfM(m)I(\bfX_m;\bfZ'_m) = \sum_{m}P_\bfM(m)I(\bfX;\bfZ'|\bfM=m) = I(\bfX;\bfZ'|\bfM).
    \end{align*}
    
    Continuing the calculation, 
    \begin{align*}
        I(\bfM;\bfZ)    & \leq nC(\cW_E) - I(\bfX;\bfZ'|\bfM)\\
        & = I(\bfX;\bfZ') - I(\bfX;\bfZ'|\bfM) + nC(\cW_E)-I(\bfX;\bfZ')\\
        & = I(\bfM,\bfX;\bfZ') - I(\bfX;\bfZ'|\bfM) + nC(\cV_E)-I(\bfX;\bfZ')\\
        & = I(\bfM;\bfZ') + nC(\cV_E)-I(\bfX;\bfZ').
    \end{align*}
    Since $H(\bfZ') =  H(\bfZ'|\bfX) + I(\bfX;\bfZ')$, we have 
    \begin{align*}
        D(P_{\bfZ'}\|P_{\bfU_n}) &= n-H(\bfZ') = n - H(\bfZ'|\bfX)-I(\bfX;\bfZ') \\
        &= nC(\cV_E)-I(\bfX;\bfZ').
    \end{align*}
    This yields 
    \begin{align*}
        I(\bfM;\bfZ) &\leq I(\bfM;\bfZ') + D(P_{\bfZ'}\|P_{\bfU_n})\\
        &\stackrel{\eqref{eq: leakage}}\leq  D(\cV_E^n \circ P_{\mC_E}\|P_{\bfU_n})+ D(P_{\bfZ'}\|P_{\bfU_n}).
    \end{align*}
The proof will be complete once we show that $D(P_{\bfZ'}\|P_{\bfU_n})\le D(\cV_E^n \circ P_{\mC_E}\|P_{\bfU_n})$.
Let $\bfV$ be the additive noise (a random Bernoulli vector) corresponding to the BSC $\cV_E$ and let $\bfX' \sim_U \mC_E$ be independent of $\bfV$. Further, define $\bfT$ to be a uniform random element from the set of coset leaders in $\mC_B/\mC_E$, chosen independently of the other random quantities that are involved in the problem
description. We can write $\bfX = \bfX' + \bfT$ and $\bfZ' = \bfX + \bfV$. This yields that $P_\bfX  = P_{\mC_B} = P_{\mC_E} \ast P_{\bfT}$, and  $P_{\bfZ'} = P_{\bfV} \ast P_{\mC_B}$ 
Using these relations, we finally obtain
    \begin{align*}
        D(P_{\bfZ'}\|P_{\bfU_n}) & =  D(P_{\bfT} \ast P_{\bfV} \ast P_{\mC_E} \| P_{\bfU_n}) \\
        & = D(P_{\bfT} \ast P_{\bfV} \ast P_{\mC_E} \| P_{\bfT} \ast P_{\bfU_n})  \quad (\text{because\ } 
 P_{\bfU_n} = P_{\bfT} \ast P_{\bfU_n})\\
        & \leq D(P_{\bfV} \ast P_{\mC_E} \| P_{\bfU_n})  = D(\cV_E^n \circ P_{\mC_E}\|P_{\bfU_n}),
    \end{align*}
    where the inequality above follows by the data processing inequality.
\end{proof}

The above result implies that if a code pair $\mC_E\subset\mC_B$ can operate on a BMS-BSC wiretap channel, secrecy is preserved when the BSC 
$\cV_E$ is replaced by a general BMS channel $\cW_E$. Moreover, if $\cW_E$ is less noisy than $\cW_B$, the
communication rate on the main channel approaches the capacity value $C_s(\cW_B,\cW_E) = C(\cW_B)-C(\cW_E)$. 
\begin{theorem} With the assumptions of Theorem~\ref{thm: wtc}, its conclusions are true if the BSC $\cV_E$ is replaced
with an arbitrary BMS channel $\cW_E$. If $\cW_E$ is less noisy than $\cW_B$, then the communication rate 
satisfies $R(\mC_B,\mC_E)\ge C_s(\cW_B,\cW_E)-\epsilon$. 
\end{theorem}

This theorem establishes our main result for LDPC codes on the BMS wiretap channel. An intriguing question is 
proving that these codes support transmission with strong secrecy. Another problem is showing that 
the $\log^2 n$ information leakage is attainable with ensembles of irregular LDPC codes that allow
efficient decoding at rates close to secrecy capacity. 

\section*{Appendix: Proof of Proposition.~\ref{prop: main}}\label{appendixA}

When proving Proposition \ref{prop: main}, we use the following two generic lemmas.
\begin{lemma}\label{lem: A1}
    Let $a_1,\dots,a_l$ be a sequence of non-negative reals and let $s \in (0,1)$. Then
    \begin{align*}
        (a_1 + \dots + a_l)^s \leq a_1^s + \dots + a_l^s.
    \end{align*}
\end{lemma}
\begin{proof}
    Let $a = (a_1, \dots, a_l)$. Then the above expression is equivalent to $\|a\|_1 \leq \|a\|_{s}$ which is true by the monotonicity of norms (the norms are understood with respect to the counting measure).
\end{proof}
The next lemma follows by the rearrangement inequality \cite[Sec. 10.2]{hardy1952inequalities}, which says that, given two sequences of nonnegative numbers $(a_1,\dots,a_n)$ and $(b_1,\dots,b_n)$, the sum of permuted products $\sum_{i}a_{\tau(i)}b_{\sigma(i)}$ is maximized
when both sequences are arranged in a monotone (nondecreasing or nonincreasing) order.
\begin{lemma}\label{lem: A2}
    Let $s>0$ and let $a_1,\dots,a_l$ be a sequence of non-negative reals. Let $b_1,\dots,b_l$ be a permutation of $a_1^s,\dots,a_l^s$. Then 
    \begin{align*}
        \sum_{i=1}^l a_ib_i \leq \sum_{i=1}^l a_i^{s+1}.
    \end{align*}
\end{lemma}

We proceed to prove Proposition~\ref{prop: main}:
\begin{align}
    & 2^{(\alpha-1)D_{\alpha}(P_{\bfH \bfV, \bfH}\|P_{\bfU_m}P_\bfH)} = \hspace*{-.2in} \sum_{u \in \ff_2^m,H \in \mathcal H}\hspace*{-.05in}\frac{P_{\bfH \bfV, \bfH}(u,H)^\alpha}{P_{\bfU_m}(u)^{\alpha-1}P_\bfH(H)^{\alpha-1}}\nonumber\\
    &= 2^{m(\alpha-1)}\sum_{H}P_\bfH(H)\sum_{u}(P_{\bfH \bfV|\bfH}(u|H))^{\alpha} \nonumber\\
    &= 2^{m(\alpha-1)}\sum_{H}P_\bfH(H)\sum_{u}\Big(\sum_{v \in \ff_2^n}P_{\bfH \bfV, \bfV|\bfH}(u,v|H)\Big)^{\alpha}\nonumber\\
    &= 2^{m(\alpha-1)}\sum_{H}P_\bfH(H)\sum_{u}\sum_{v}P_{\bfH \bfV, \bfV|\bfH}(u,v|H) \Big(\sum_{v' \in \ff_2^n}P_{\bfH \bfV,\bfV|\bfH  }(u,v'|H)\Big)^{\alpha-1}\nonumber\\
    &= 2^{m(\alpha-1)}\sum_{H}P_\bfH(H)\sum_{u}\sum_{v}P_{\bfH \bfV, \bfV|\bfH}(u,v|H) \Big(\hspace*{-.1in}\sum_{v' \in \Gamma(v,t)}\hspace*{-.1in}P_{\bfH \bfV,\bfV|\bfH  }(u,v'|H) 
    +  \hspace*{-.2in}\sum_{v' \in \ff_2^n \setminus \Gamma(v,t)} \hspace*{-.2in}P_{\bfH \bfV,\bfV|\bfH  }(u,v'|H)\Big)^{\alpha-1}\nonumber\\
    &\le 2^{m(\alpha-1)}\sum_{H}P_\bfH(H)\sum_{u}\sum_{v}P_{\bfH \bfV, \bfV|\bfH}(u,v|H)\Big(\sum_{v' \in \Gamma(v,t)}P_{\bfH \bfV,\bfV|\bfH  }(u,v'|H)\Big)^{\alpha-1} \nonumber\\
    &\hspace*{.3in}+  2^{m(\alpha-1)}\sum_{H}P_\bfH(H)\sum_{u}\sum_{v}P_{\bfH \bfV, \bfV|\bfH}(u,v|H) \Big(\sum_{v' \in \ff_2^n \setminus \Gamma(v,t)}P_{\bfH \bfV,\bfV|\bfH  }(u,v'|H)\Big)^{\alpha-1},
    \label{eq:TwoTerms}
\end{align}
where the inequality on the last line follows from Lemma \ref{lem: A1}.

We shall bound the two terms on the right-hand side of \eqref{eq:TwoTerms} separately.
Bounding the first term, 
\begin{align*}
     2^{m(\alpha-1)}&\sum_{H}P_\bfH(H)\sum_{u}\sum_{v}P_{\bfH \bfV, \bfV|\bfH}(u,v|H)\Big(\sum_{v' \in \Gamma(v,t)}P_{\bfH \bfV,\bfV|\bfH  }(u,v'|H)\Big)^{\alpha-1} \\
    &= 2^{m(\alpha-1)}\sum_{H}P_\bfH(H)\sum_{u}\sum_{v}P_{\bfH \bfV, \bfV|\bfH}(u,v|H)\Big(\sum_{e \in \Gamma(0,t)}P_{\bfH \bfV, \bfV|\bfH}(u,v+e|H)\Big)^{\alpha-1} \\
    &\leq  2^{m(\alpha-1)}\sum_{H}P_\bfH(H)\sum_{u}\sum_{v}P_{\bfH \bfV, \bfV|\bfH}(u,v|H)\sum_{e \in \Gamma(0,t)}P_{\bfH \bfV, \bfV|\bfH}(u,v+e|H)^{\alpha-1} \hspace*{.2in} (\text{Lemma \ref{lem: A1}})\\
    &= 2^{m(\alpha-1)}\sum_{H}P_\bfH(H)\sum_{u}\sum_{e \in \Gamma(0,t)}\sum_{v}P_{\bfH \bfV, \bfV|\bfH}(u,v|H)P_{\bfH \bfV, \bfV|\bfH}(u,v+e|H)^{\alpha-1} \\
    &\leq  2\nu_t 2^{m(\alpha-1)}\sum_{H}P_\bfH(H)\sum_{u}\sum_{v}P_{\bfH \bfV, \bfV|\bfH}(u,v|H)^\alpha \quad (\text{Lemma \ref{lem: A2}})\\
     &= 2\nu_t 2^{m(\alpha-1)}\sum_{H}P_\bfH(H)\sum_{u}\sum_{v}P_\bfV(v)^{\alpha}P_{\bfH \bfV|\bfH, \bfV}(u|H,v )^{\alpha}\\
    &= 2\nu_t 2^{m(\alpha-1)}\sum_{v}P_\bfV(v)^{\alpha}\sum_{H}P_\bfH(H)\sum_{u}\1\{H v=u\} \\
    &= 2\nu_t 2^{m(\alpha-1)}\sum_{v}P_\bfV(v)^{\alpha}\sum_{H}P_\bfH(H) \\
    &= 2\nu_t 2^{m(\alpha-1)}\sum_{v}P_\bfV(v)^{\alpha}\\
    &= 2\nu_t 2^{(\alpha-1)(m-H_\alpha(\bfV))}.
\end{align*}

Bounding the second term in \eqref{eq:TwoTerms}: 
\begin{align*}
    2^{m(\alpha-1)}&\sum_{H}P_\bfH(H)\sum_{u}\sum_{v}P_{\bfH \bfV, \bfV|\bfH}(u,v|H)\Big(\sum_{v' \in \ff_2^n \setminus \Gamma(v,t)}P_{\bfH \bfV,\bfV|\bfH  }(u,v'|H)\Big)^{\alpha-1}\\
    &\le 2^{m(\alpha-1)}\Big(\sum_{H}P_\bfH(H)\sum_{u}\sum_{v}P_{\bfH \bfV, \bfV|\bfH}(u,v|H)\sum_{v' \in \ff_2^n \setminus \Gamma(v,t)}P_{\bfH \bfV,\bfV|\bfH  }(u,v'|H)\Big)^{\alpha-1} \qquad\text{(Jensen's)}\\
    &= 2^{m(\alpha-1)}\Big(\sum_{v}\sum_{v' \in \ff_2^n \setminus \Gamma(v,t)} P_\bfV(v)P_\bfV(v')\sum_{u}\sum_{H}P_\bfH(H)P_{\bfH \bfV|\bfH \bfV}(u|H, v )P_{\bfH \bfV|\bfH, \bfV }(u|H, v')\Big)^{\alpha-1}\\
    &= 2^{m(\alpha-1)}\Big(\sum_{v}\sum_{v' \in \ff_2^n \setminus \Gamma(v,t)} P_\bfV(v)P_\bfV(v')\sum_{u}\sum_{H}P_\bfH(H)\1\{H v= H v'=u\}\Big)^{\alpha-1}\\
    &= 2^{m(\alpha-1)}\Big(\sum_{v}\sum_{v' \in \ff_2^n \setminus \Gamma(v,t)} P_\bfV(v)P_\bfV(v')\sum_{u}\Pr(\bfH v= \bfH v' =u)\Big)^{\alpha-1}\\
    &= 2^{m(\alpha-1)}\Big(\sum_{v}\sum_{v' \in \ff_2^n \setminus \Gamma(v,t)} P_\bfV(v)P_\bfV(v')\Pr(\bfH v=\bfH v')\Big)^{\alpha-1}\\
    &= \Big(2^m\sum_{v}\sum_{v' \in \ff_2^n \setminus \Gamma(v,t)} P_\bfV(v)P_\bfV(v')\Pr(\bfH (v-v') = 0)\Big)^{\alpha-1}\\
    &= \Big(2^m\sum_{v}\sum_{v'' \in \ff_2^n \setminus \Gamma(0,t)} P_\bfV(v)P_\bfV(v+v'')\Pr(\bfH v'' = 0)\Big)^{\alpha-1}\\
    &= \Big(2^m\sum_{v'' \in \ff_2^n \setminus \Gamma(0,t)} \Pr(\bfH v'' = 0) \sum_{v}P_\bfV(v)P_\bfV(v+v'')\Big)^{\alpha-1}\\
    &= \Big(2^m\sum_{v'' \in \ff_2^n \setminus \Gamma(0,t)} \Pr(\bfH v'' = 0) P_\bfV \ast P_\bfV(v'')\Big)^{\alpha-1}.
\end{align*}

Combining the two inequalities just derived, we obtain the bound \eqref{eq: main}.

\bibliographystyle{abbrv}
\bibliography{smoothing}
\end{document}